\documentclass{article}

\usepackage{macro}

\usepackage[most]{tcolorbox}

\AtEveryBibitem{%
    \clearname{editor}%
    \clearname{editora}%
    \clearname{editorb}%
    \clearname{editorc}%
}

\hypersetup{
    colorlinks=true,
    linkcolor=purple,
    citecolor=forestgreen,
}

\title{Capacity of Uniform Noise Channels\\Under Average Input Power Constraints}

\author{
    Yihan Zhang\thanks{
        School of Mathematics, University of Bristol. 
        Email: \href{mailto:yihan.zhang@bristol.ac.uk}{\texttt{yihan.zhang@bristol.ac.uk}}.
    }
}

\begin{document}
\maketitle

\begin{abstract}
The foundational work of Shannon (1948) identified the capacity of an additive noise channel under an average input power constraint as a mutual information maximization problem over input densities subject to a second moment constraint. However, a quantitative understanding of the channel capacity is significantly lacking even for very simple noise distributions beyond Gaussians. In particular, it is a long standing question to determine the capacity of channels with noise uniformly distributed over a centered interval. This paper settles this question by precisely characterizing the capacity and the corresponding capacity achieving input and output distributions of such channels. A key observation en route to these results is a certain periodization identity for the output density of a uniform noise channel which in turn allows for applications of Fourier analytic techniques. 
\end{abstract}



\newcounter{asmpctr} 
\setcounter{asmpctr}{\value{enumi}}


\section{Introduction}

The foundational work of Shannon \cite{Shannon} identified the capacity of a scalar additive channel 
\begin{align}
    Y &= X + Z \label{eqn:channel}
\end{align}
where the input $ X $ and the noise $ Z $ are independent of each other. 
Its capacity under an average input power constraint $\ic>0$ is given by
\begin{align}
    C(\ic,f_Z) &\coloneqq \sup_{f_X : \expt{ X^2 } \le \ic} h(X + Z) - h(Z) , \label{eqn:cap} 
\end{align}
where $ f_X,f_Z $ denote the densities of $X,Z$, respectively. 
When $Z$ is Gaussian, the capacity achieving input (and therefore output) distribution is also Gaussian. 
If one replaces the \emph{average} power constraint $ \expt{X^2}\le\ic $ with an \emph{amplitude} power constraint $ \abs{X}\le\ic $ almost surely, then the capacity achieving input distribution can be discrete or have finite support. 
Understanding the qualitative features of such distributions has been a subject of active research, see e.g.\ \cite{Smith,McKellips,Sharma_Shamai,Huang_Meyn,Thangaraj_Kramer_Bocherer,Dytso_Yagli_Poor_Shamai,Wang_Barletta_Dytso,Barletta_Dytso}. 
Perhaps surprisingly, the capacity of channels with non-Gaussian noise is generally rather poorly understood. 
The special case where $Z$ follows the uniform distribution on an interval and $X$ obeys an amplitude power constraint was solved in \cite[Exercise 7.5]{Gallager} and subsequent studies include \cite{Oettli,Rioul_Magossi,Stapmanns_Dias_Eilers_Pfister,Stapmanns_Eilers_Dias_Kuhn_Pfister}. 
Unfortunately, this line of work has no bearing on channel capacity in the practically more pertinent case of uniform noise and \emph{average} power constraint. 
This question was explicitly labeled as open in \cite[Section VII]{Madiman_Nayar_Tkocz} and some numerical studies can be found in \cite{Letizia_Tonello}. 
There is a sizable literature on estimating the capacity of general additive channels \cite{Shannon,Ihara,Binia,Diggavi_Cover,Zamir_Erez,Verdu,Egan_Perlaza,Egan_Perlaza_Kungurtsev}.
None of these bounds, upon specialized to our setting here, is sharp. 
Many existing results on general additive channels require various regularities of the noise density such as 
tail positivity \cite{Das}, 
everywhere positivity \cite{Fahs_Ajeeb_Abou-Faycal}, 
analyticity \cite{Tchamkerten}, 
analytic extendability \cite{Fahs_Abou-Faycal}, 
Gaussian smoothing \cite{Ranjbar_Tran_Nguyen_Gursoy_Nguyen-Le}, etc. 
The uniform distribution, albeit being very simple, does not possess any of such properties and therefore resists analyses of those types. 
These obstructions were also observed in \cite{Madiman_Nayar_Tkocz}. 

The aim of this paper is to offer an explicit expression of the capacity and the corresponding capacity achieving input and output distributions for the additive channel \Cref{eqn:channel} with noise uniformly distributed on a compact interval centered around the origin. 
The key observation that unlocks these precise characterizations is that for $ Z \sim \unif([-\nc,\nc)) $ and any $X$ independent of $Z$, the periodization of the output density $ f_Y $ is a constant. 
Specifically, consider the lattice $ \Lambda_\nc \coloneqq 2\nc\bbZ = \brace{ 2\nc k : k\in\bbZ } $ and the periodization $ \sfP_{\Lambda_\nc} f_Y $ of $ f_Y $ with respect to $\Lambda_\nc$ defined as the series given by summing the translates of $ f_Y $ by elements in $\Lambda_\nc$: 
\begin{align}
&&
    \sfP_{\Lambda_\nc} f_Y(s) &\coloneqq \sum_{\tau \in \Lambda_\nc} f_Y(s+\tau) , & 
    s &\in [-\nc,\nc) . & 
& \notag 
\end{align}
We show in \Cref{itm:fY_sum} of \Cref{lem:Y} that $ \sfP_{\Lambda_\nc} f_Y $ is identically equal to the constant $ 1/(2\nc) $ on $ \bbR/\Lambda_\nc \cong [-\nc,\nc) $. 
This admits a probabilistic interpretation. 
Denoting by $ [y]_\nc \in [-\nc,\nc) $ the residue of $y\in\bbR$ modulo $ \Lambda_\nc $ and noting that the density of $ [Y]_\nc $ is precisely $ \sfP_{\Lambda_\nc} f_Y $, we have that $ [Y]_\nc \sim \unif([-\nc,\nc)) $. 
This crucial property is satisfied by the output of a uniform noise channel with \emph{any} input. 
Taking this into account allows us to explicitly solve the mutual information maximization problem in \Cref{eqn:cap} which is equivalent to maximizing the output differential entropy. 
Moreover, we show in \Cref{lem:Flam} that the resulting unique capacity achieving input distribution is absolutely continuous (with respect to Lebesgue measure on the real line), in contrast to discrete structures arising from amplitude power constraints \cite{Tchamkerten,Fahs_Abou-Faycal}. 
To the best of our knowledge, the identified capacity achieving input and output distributions are unique to the specific problem at hand and do not seem to naturally arise elsewhere in probability theory. 

\section{Main results}

\paragraph{Notation.}
For a probability density function $f$, we write $ X \sim f $ to mean that the distribution of $X$ has density $f$. 
We denote by $ h(f) $ or $ h(X) $ the differential entropy of $f$, and by $ D(f\Mid g) $ the Kullback--Leibler divergence between $ f,g $. 
All logarithms are to the base $e$. 

We consider the additive channel \Cref{eqn:channel} with $ Z \equiv Z_\nc \sim \unif([-\nc,\nc)) $ where $ \nc>0 $ is fixed. 
In this case, we slightly abuse notation and use $ C(\ic,\nc) $ to denote the corresponding channel capacity \Cref{eqn:cap}. 
Clearly, $ C(\ic,\nc) $ depends on the input and power constraints $ \ic $ and $ \nc $ only through the signal-to-noise ratio $ \ic/\nc^2 $ (noting that $ \expt{Z_\nc^2} = \nc^2/3 $) and one can without loss of generality set $ \ic=1 $. 
However, we keep both parameters generic, following the convention in information theory. 

The main result of this paper, \Cref{thm:cap} below, identifies the capacity and the unique capacity achieving input and output distributions for uniform noise channels under average input power constraints. 
Moreover, the input and output densities are shown to be absolutely continuous. 
For a formal statement of the results, a few definitions are in order. 
For $ y\in\bbR $, let $ [y]_\nc $ be the unique $ s\in [-\nc,\nc) $ such that $ y - s \in 2\nc\bbZ $. 
For $ \lambda>0 $ and $ s\in [-\nc,\nc) $, define 
\begin{align}
    \Theta_{\lambda,\nc}(s) &\coloneqq \sum_{k\in\bbZ} \exp\paren{ -\lambda(s + 2k\nc)^2 } , \label{eqn:Theta} \\
    Q_\nc(\lambda) &\coloneqq 
    - \frac{1}{2\nc} \int_{-\nc}^\nc \frac{\partial}{\partial\lambda} \log( \Theta_{\lambda,\nc}(s) ) \diff s
    = \frac{1}{2\nc} \int_{-\nc}^\nc \frac{1}{\Theta_{\lambda,\nc}(s)} \sum_{k\in\bbZ} (s+2k\nc)^2 \exp\paren{-\lambda (s+2k\nc)^2} \diff s . \label{eqn:Q} 
\end{align}
As shown in \Cref{lem:Q} below, $ Q_\nc $ is a continuous and strictly decreasing function with range $ (\nc^2/3,\infty) $. 
Therefore, for every $ \ic>0 $, the equation 
\begin{align}
    Q_\nc(\lambda) = \ic + \nc^2/3 \label{eqn:lambda}
\end{align}
has a unique solution $ \lambda \equiv \lambda(\ic,\nc)>0 $. 

\begin{theorem}
\label{thm:cap}
Fix $ \ic,\nc>0 $. 
Consider the channel \Cref{eqn:channel} with noise $ Z \equiv Z_\nc \sim \unif([-\nc,\nc)) $ and average input power constraint $ \expt{X^2} \le \ic $. 
Let $ \lambda \equiv \lambda(\ic,\nc) $ be defined through \Cref{eqn:lambda}. 

\begin{enumerate}
    \item \label{itm:thm1} For any $ y\in\bbR $, let
    \begin{align}
        g_{\lambda,\nc}(y) &\coloneqq \frac{\exp\paren{-\lambda y^2}}{2\nc \Theta_{\lambda,\nc}([y]_\nc)} . \label{eqn:g} 
    \end{align}
    Then $ g_{\lambda,\nc} $ is the unique capacity achieving output density. 

    \item \label{itm:thm2} The channel capacity equals 
    \begin{align}
        C(\ic,\nc) &= \lambda (\ic + \nc^2/3) + \frac{1}{2\nc} \int_{-\nc}^\nc \log(\Theta_{\lambda,\nc}(s)) \diff s . \notag 
    \end{align}

    \item \label{itm:thm3} There exists a unique capacity achieving input distribution that is  absolutely continuous with density $ f_X $ defined as follows. 
    Given $ x\in\bbR $, there exists a unique pair $ (s,K)\in[-\nc,\nc)\times\bbZ $ such that $ x - \nc = s + 2K\nc $. 
    Then 
    \begin{align}
        f_X(x) &= \frac{4\lambda \nc}{\Theta_{\lambda,\nc}(s)^2} \sum_{\substack{k\le K \\ j>K}} (j-k) \exp\paren{ -\lambda\brack{(s + 2k\nc)^2 + (s + 2j\nc)^2} } . \label{eqn:fX} 
    \end{align}
    In particular, $ f_X $ is a probability density function satisfying $ \expt{X^2} = \ic $ and $ X + Z_\nc \sim g_{\lambda,\nc} $ where $ X \sim f_X $ is independent of $ Z_\nc $. 
\end{enumerate}
\end{theorem}

The unique capacity achieving input and output distributions $ f_X $ in \Cref{eqn:fX} and $ g_{\lambda,\nc} $ in \Cref{eqn:g} are plotted in \Cref{fig:fg} for several values of $ \nc $, assuming $ \ic = 1 $. 
These densities appear visually Gaussian like. 
This is not surprising given the result of \cite{Zamir_Erez} saying that for any additive channel under an average power constraint, a Gaussian input achieves a rate no more than $0.5$ bits below capacity. 
Also, \cite{Madiman_Nayar_Tkocz} showed that the capacity of any additive channel with symmetric log-concave noise (including the uniform noise) under an average power constraint is at most $0.254$ bits larger than the capacity of a Gaussian channel (whose capacity achieving input distribution is well known to be Gaussian) with matching noise variance. 
The above being said, the significance of \Cref{thm:cap} is that it characterizes the precise input and output densities which are not exactly Gaussians. 

\begin{figure}[htbp]
    \centering
    \begin{subfigure}{0.49\textwidth}
        \centering
        \includegraphics[width=\textwidth]{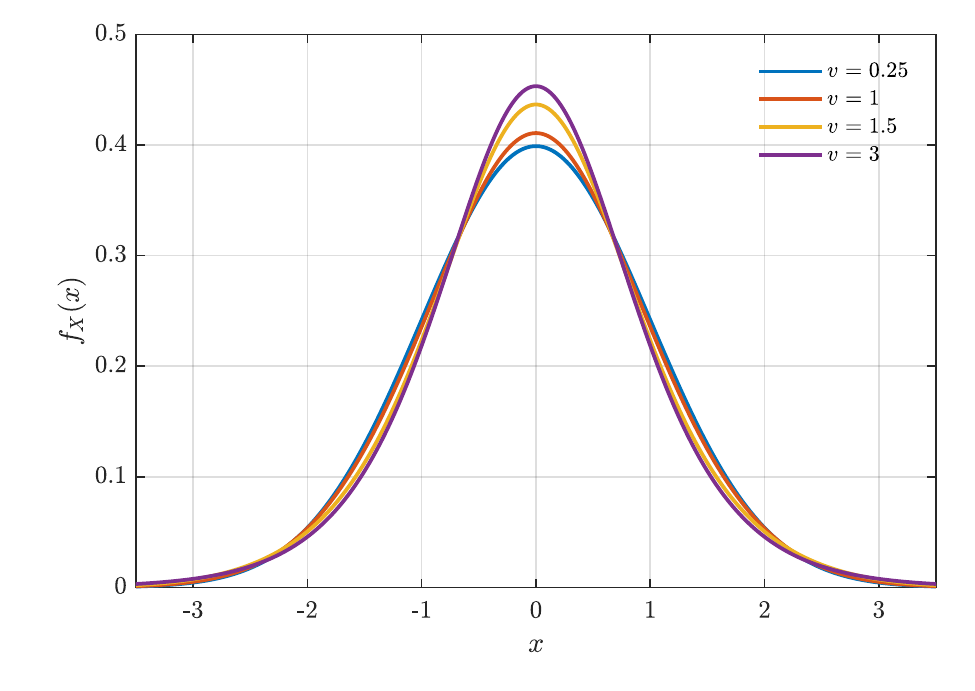}
    \end{subfigure}
    \begin{subfigure}{0.49\textwidth}
        \centering
        \includegraphics[width=\textwidth]{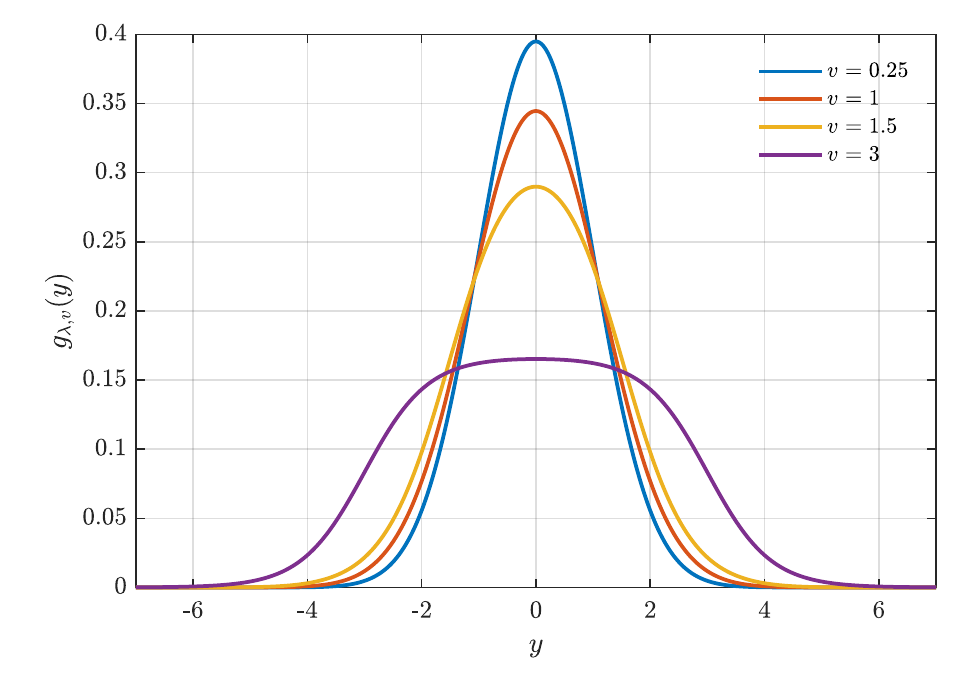}
    \end{subfigure}
    \caption{Plots of the unique capacity achieving input and output distributions $ f_X $ in \Cref{eqn:fX} and $ g_{\lambda,\nc} $ in \Cref{eqn:g} for $ \nc\in\brace{0.25,1,1.5,3} $, assuming $ \ic=1 $.}
    \label{fig:fg}
\end{figure}

\section{Discussion}

This paper settles the long standing question of determining the capacity and capacity achieving input and output distributions of channels with additive uniform noise under average power constraints.
The input and output distributions identified in \Cref{thm:cap} are shown to be absolutely continuous. 
We leave it for future work to study their qualitative features such as higher-order regularities, Gaussian approximability, other extremal properties, etc. 

More broadly, we still do not have a satisfactory quantitative understanding of the capacity and capacity achieving input and output distributions of general additive channels. 
We take this opportunity to repeat the question raised in \cite[Question 1]{Madiman_Nayar_Tkocz}: among all symmetric log-concave noise distributions with fixed variance, which distributions maximize the capacity of the corresponding additive channels?
The authors of \cite{Madiman_Nayar_Tkocz} showed that the uniform distribution minimizes entropy among all symmetric log-concave distributions and speculated that it is a candidate extremizer for the above question. 
We hope that the detailed study of uniform noise channels in this paper can drive us closer to the answer of this question. 

\section*{Acknowledgements}
The author thanks Tomasz Tkocz for making him aware of the tantalizing question of determining the capacity of uniform noise channels in 2018. 

\section{Proofs}
\label{sec:pf}

We first collect a few observations on the convolution of any input density $ f_X $ with the uniform density $ f_{Z_\nc} $. 
The simple yet crucial periodization identity \Cref{eqn:fY} will be exploited in the remainder of the proofs. 

\begin{lemma}
\label{lem:Y}
Let $ Y = X + Z_\nc $ where $ Z_\nc \sim \unif([-\nc,\nc)) $ and $ X $ is a real-valued random variable, independent of $ Z_\nc $, with cumulative distribution function $ F_X $. 
Then the following results hold. 
\begin{enumerate}
    \item \label{itm:fY} The density of $Y$ is given by 
    \begin{align}
        f_Y(y) &= \frac{1}{2\nc} \paren{ F_X(y+\nc) - F_X(y-\nc) } . \notag 
    \end{align}

    \item \label{itm:fY_sum} For any $ s\in[-\nc,\nc) $, 
    \begin{align}
        \sum_{k\in\bbZ} f_Y(s + 2k\nc) &= \frac{1}{2\nc} . \label{eqn:fY} 
    \end{align}

    \item \label{itm:EY2} If we further assume that $X$ has finite second moment, then 
    \begin{align}
        \expt{Y^2} &= \expt{X^2} + \nc^2/3 . \notag 
    \end{align}
\end{enumerate}
\end{lemma}

\begin{proof}
For \Cref{itm:fY}, 
\begin{align}
    f_Y(y)
    &= \expt{ f_{Z_\nc}(y - X) }
    = \frac{1}{2\nc} \expt{\one_{[-\nc,\nc)}(y - X)} \notag \\
    &= \frac{1}{2\nc} \prob{ y-\nc < X \le y+\nc }
    = \frac{1}{2\nc} (F_X(y+\nc) - F_X(y-\nc)) . \notag 
\end{align}

For \Cref{itm:fY_sum}, let us first use \Cref{itm:fY} to compute the sum truncated at $ \pm N $ for some $ N\in\bbZ_{>0} $:  
\begin{align}
    \sum_{k=-N}^N f_Y(s + 2k\nc)
    &= \frac{1}{2\nc} \sum_{k=-N}^N \paren{ F_X(s+2k\nc + \nc) - F_X(s+2k\nc - \nc) } \notag \\
    &= \frac{1}{2\nc} \sum_{k=-N}^N \paren{ F_X(s+2(k+1)\nc - \nc) - F_X(s+2k\nc - \nc) } \notag \\
    &= \frac{1}{2\nc} \paren{ F_X(s+2(N+1)\nc - \nc) - F_X(s-2N\nc - \nc) } , \notag 
\end{align}
since it is a telescoping sum. 
Now sending $ N\to\infty $, we have $ F_X(s+2(N+1)\nc - \nc) \to 1 $ and $ F_X(s-2N\nc - \nc) \to 0 $, which establishes the result. 

Finally, \Cref{itm:EY2} immediately follows from independence between $ X,Z_\nc $ and the fact $ \expt{Z_\nc^2} = \nc^2/3 $. 
\end{proof}

The next lemma establishes certain regularity properties of the function $ Q_\nc $. 
In particular, it implies the existence and uniqueness of solution to the equation \Cref{eqn:lambda} for any $ \ic>0 $. 

\begin{lemma}
\label{lem:Q}
For any $\nc>0$, the function $ Q_\nc $ defined in \Cref{eqn:Q} is continuous and strictly decreasing on $ (0,\infty) $, and satisfies 
\begin{align}
&&
    \lim_{\lambda\downarrow0} Q_\nc(\lambda) &= \infty , & 
    \lim_{\lambda\to\infty} Q_\nc(\lambda) &= \nc^2/3 . & 
& \notag 
\end{align}
\end{lemma}

\begin{proof}
Denote 
\begin{align}
&&
    a_k(s) &\coloneqq s + 2k\nc , & 
    p_{k,\lambda}(s) &\coloneqq \frac{1}{\Theta_{\lambda,\nc}(s)} \exp\paren{ -\lambda a_k(s)^2 } , &
& \label{eqn:pk} 
\end{align}
where we suppress the dependence on $ \nc $ in notation since $ \nc $ is fixed throughout. 
Then $ p_{\cdot,\lambda}(s) $ is the probability mass function of a probability distribution on $ \bbZ $. 
A straightforward calculation shows
\begin{align}
    \frac{\partial}{\partial\lambda} p_{k,\lambda}(s)
    &= - \brack{
        \frac{a_k(s)^2 e^{-\lambda a_k(s)^2}}{\Theta_{\lambda,\nc}(s)}
        - \frac{e^{-\lambda a_k(s)^2} \sum_{j\in\bbZ} a_j(s)^2 e^{-\lambda a_j(s)^2}}{\Theta_{\lambda,\nc}(s)^2}
    } \notag \\
    &= - \brack{
        a_k(s)^2 p_{k,\lambda}(s)
        - p_{k,\lambda}(s) \sum_{j\in\bbZ} a_j(s)^2 p_{j,\lambda}(s)
    } . \notag 
\end{align}
Using this, we have
\begin{align}
    \frac{\partial}{\partial\lambda} \sum_{k\in\bbZ} a_k(s)^2 p_{k,\lambda}(s)
    = - \brack{
        \sum_{k\in\bbZ} a_k(s)^4 p_{k,\lambda}(s)
        - \paren{
            \sum_{k\in\bbZ} a_k(s)^2 p_{k,\lambda}(s)
        }^2
    }
    = - \var{ a_\bk(s)^2 } , \notag 
\end{align}
where $ \bk \sim p_{\cdot,\lambda}(s) $. 
This quantity is strictly negative for every $ s\in[-\nc,\nc) $. 
Since $ Q_\nc $ can be written as 
\begin{align}
    Q_\nc(\lambda) &= \frac{1}{2\nc} \int_{-\nc}^\nc \sum_{k\in\bbZ} a_k(s)^2 p_{k,\lambda}(s) \diff s , \notag 
\end{align}
we conclude the monotonicity of $ Q_\nc $. 

Next, we prove continuity of $ Q_\nc $. 
Recall $ \Theta_{\lambda,\nc} $ from \Cref{eqn:Theta} and define also 
\begin{align}
    \Omega_{\lambda,\nc}(s) &\coloneqq \sum_{k\in\bbZ} (s+2k\nc)^2 e^{-\lambda(s+2k\nc)^2} . \notag 
\end{align}
Then 
\begin{align}
    Q_\nc(\lambda) &= \frac{1}{2\nc} \int_{-\nc}^\nc \frac{\Omega_{\lambda,\nc}(s)}{\Theta_{\lambda,\nc}(s)} \diff s . \label{eqn:QQ} 
\end{align}
Note that for $ s\in[-\nc,\nc) $ and $ \abs{k}\ge1 $, we have $ \abs{s + 2k\nc} \ge \nc (2\abs{k} - 1) $. 
Consider a compact interval $ [a,b] \subset (0,\infty) $. 
Then for all $ \lambda\in[a,b] $ and $ \abs{k}\ge1 $, 
\begin{align}
    e^{-\lambda(s+2k\nc)^2} &\le \exp\paren{ - a \nc^2 (2\abs{k} - 1)^2 } , \notag 
\end{align}
the right-hand side of which is summable over $ k\in\bbZ $. 
Therefore, $ \Theta_{\lambda,\nc}(s) $ converges uniformly in $ (\lambda,s) \in [a,b] \times [-\nc,\nc) $. 
A similar argument shows that $ \Omega_{\lambda,\nc}(s) $ also uniformly converges. 
Moreover, since each summand in $ \Theta_{\lambda,\nc}(s) $ and $ \Omega_{\lambda,\nc}(s) $ is continuous in $ (\lambda,s) $, $ \Theta_{\lambda,\nc}(s) $ and $ \Omega_{\lambda,\nc}(s) $ are also continuous. 
The division by $ \Theta_{\lambda,\nc}(s) $ in \Cref{eqn:QQ} is legit since $ \Theta_{\lambda,\nc}(s) $ is bounded away from zero uniformly over all $ (\lambda,s)\in[a,b]\times[-\nc,\nc) $. 
The latter can be seen by only looking at the term $k=0$: 
\begin{align}
    \Theta_{\lambda,\nc}(s) &\ge e^{-\lambda s^2} \ge e^{-b\nc^2} > 0 . \notag 
\end{align}
Therefore, $ (\lambda,s) \mapsto \Omega_{\lambda,\nc}(s) / \Theta_{\lambda,\nc}(s) $ is continuous and bounded uniformly on $ [a,b] \times [-\nc,\nc) $. 
Integrating this over $ s\in[-\nc,\nc) $ implies that $ Q_\nc $ is continuous on $ (0,\infty) $. 

As $ \lambda\to\infty $, the mass of $ p_{\cdot,\lambda} $ concentrates on points in the lattice $ \brace{a_k(s) : k\in\bbZ} = s + 2\nc \bbZ $ whose distance to $ 0 $ is minimized. 
For $ s\in(-\nc,\nc) $, the minimizer is unique and is given by $a_k(s) = s$. 
If $ s=-\nc $, the minimizer is not unique, but this case can be ignored since it does not contribute to the integral in the definition \Cref{eqn:Q} of $ Q_\nc $. 
Therefore, 
\begin{align}
    \lim_{\lambda\to\infty} Q_\nc(\lambda) &= \frac{1}{2\nc} \int_{-\nc}^\nc s^2 \diff s 
    = \nc^2/3 . \notag 
\end{align}

To study the $ \lambda\downarrow0 $ limit of $ Q_\nc $, let us first estimate $ \Theta_{\lambda,\nc} $ for small $ \lambda>0 $ using Fourier series (a.k.a.\ Poisson resummation upon periodizing $ \Theta_{\lambda,\nc} $ as below). 
Extending the domain of $ \Theta_{\lambda,\nc} $ from $ [-\nc,\nc) $ to $ \bbR $, we obtain a $ 2\nc $-periodic function, i.e., $ \Theta_{\lambda,\nc}(s+2\nc) = \Theta_{\lambda,\nc}(s) $ for any $ s\in\bbR $. 
The Fourier coefficient of $ \Theta_{\lambda,\nc} $ at frequency $ n\in\bbZ $ can then be computed as 
\begin{align}
    \wh{\Theta}_{\lambda,\nc}(n) 
    &= \frac{1}{2\nc} \int_0^{2\nc} \Theta_{\lambda,\nc}(s) \exp\paren{-\frac{2\pi\ii n s}{2\nc}} \diff s \notag \\
    &= \frac{1}{2\nc} \int_0^{2\nc} \sum_{k\in\bbZ} \exp\paren{-\lambda(s+2k\nc)^2} \exp\paren{-\frac{2\pi\ii n s}{2\nc}} \diff s \notag \\
    &= \frac{1}{2\nc} \sum_{k\in\bbZ} \int_{2k\nc}^{2(k+1)\nc} \exp\paren{-\lambda x^2} \exp\paren{-\frac{2\pi\ii n}{2\nc} (x - 2k\nc)} \diff x \label{eqn:change} \\
    &= \frac{1}{2\nc} \int_{-\infty}^\infty \exp\paren{-\lambda x^2} \exp\paren{ - \ii n \frac{\pi x}{\nc} } \diff x \label{eqn:phase} \\
    &= \frac{1}{2\nc} \sqrt{\frac{\pi}{\lambda}} \exp\paren{-\frac{\pi^2 n^2}{4\lambda \nc^2}} , \label{eqn:CF}
\end{align}
where \Cref{eqn:change} is obtained from the change of variable $ s+2k\nc \mapsto x $; \Cref{eqn:phase} follows since $ \exp\paren{2\pi\ii nk} = 1 $; \Cref{eqn:CF} is by the formula of the characteristic function of $ \cN(0,\sigma^2) $ evaluated at $ n\in\bbR $:
\begin{align}
    \int_{-\infty}^\infty \frac{e^{-\frac{x^2}{2\sigma^2}}}{\sqrt{2\pi\sigma^2}} e^{-\ii n x} \diff x
    &= e^{-n^2 \sigma^2/2} . \notag 
\end{align}
Summing up all phases modulated by these coefficients, we obtain the Fourier series of $ \Theta_{\lambda,\nc} $: 
\begin{align}
    \Theta_{\lambda,\nc}(s) &= \sum_{n\in\bbZ} \wh{\Theta}_{\lambda,\nc}(n) \exp\paren{ \frac{2\pi\ii n s}{2\nc} } \notag \\
    &= \frac{1}{2\nc} \sqrt{\frac{\pi}{\lambda}} \sum_{n\in\bbZ} \exp\paren{-\frac{\pi^2 n^2}{4\lambda \nc^2}} \exp\paren{ \frac{2\pi\ii n s}{2\nc} } \notag \\
    &= \frac{1}{2\nc} \sqrt{\frac{\pi}{\lambda}} \sum_{n\in\bbZ} \exp\paren{-\frac{\pi^2 n^2}{4\lambda \nc^2}} \cos\paren{ \frac{2\pi n s}{2\nc} } \label{eqn:real} \\
    &= \frac{1}{2\nc} \sqrt{\frac{\pi}{\lambda}} \brack{ 1 + 2 \sum_{n=1}^\infty \exp\paren{-\frac{\pi^2 n^2}{4\lambda \nc^2}} \cos\paren{ \frac{\pi n s}{\nc} } } , \label{eqn:Theta_TODO} 
\end{align}
where \Cref{eqn:real} follows since $ e^{\ii\theta} = \cos(\theta) + \ii \sin(\theta) $ and the imaginary part can be removed since $ \Theta_{\lambda,\nc} $ is real-valued. 
Denoting 
\begin{align}
&&
    q_\lambda &\coloneqq \exp\paren{-\frac{\pi^2}{4\lambda\nc^2}} , & 
    R_\lambda(s) &\coloneqq 2 \sum_{n=1}^\infty q_\lambda^{n^2} \cos\paren{ \frac{\pi n s}{\nc} } , & 
& \notag 
\end{align}
let us estimate $ R_\lambda(s) $ and its derivative with respect to $\lambda$. 
Taking $ \lambda>0 $ to be sufficiently small, we can make sure $ q_\lambda \in [0,1/8] $. 
Then we have the following estimate: 
\begin{align}
    \abs{ R_\lambda(s) }
    &= 2 \abs{ \sum_{n=1}^\infty q_\lambda^{n^2} \cos\paren{ \frac{\pi n s}{\nc} } } 
    \le 2 \sum_{n=1}^\infty q_\lambda^{n^2} 
    \le 2 \sum_{n=1}^\infty q_\lambda^n
    = 2 \frac{q_\lambda}{1 - q_\lambda}
    \le 4 q_\lambda \le 1/2 , \notag 
\end{align}
and consequently, 
\begin{align}
    1 + R_\lambda(s) &\in [1/2,3/2] . \label{eqn:R} 
\end{align}
Similarly, 
\begin{align}
    \abs{\frac{\partial}{\partial\lambda} R_\lambda(s)}
    &= \frac{\pi^2}{2\nc^2\lambda^2} \abs{ \sum_{n=1}^\infty n^2 q_\lambda^{n^2} \cos\paren{\frac{\pi n s}{\nc}} }
    \le \frac{\pi^2}{2\nc^2\lambda^2} \abs{ \sum_{n=1}^\infty n^2 q_\lambda^n }
    = \frac{\pi^2}{2\nc^2\lambda^2} \frac{q_\lambda (1+q_\lambda)}{(1-q_\lambda)^3}
    \le \frac{6\pi^2}{\nc^2 \lambda^2} q_\lambda . \label{eqn:R'} 
\end{align}
Since $ \Theta_{\lambda,\nc}(s) = (2\nc)^{-1} \sqrt{\pi/\lambda} \, \paren{ 1 + R_\lambda(s) } $, using \Cref{eqn:R,eqn:R'}, we have
\begin{align}
    \frac{\partial}{\partial\lambda} \log(\Theta_{\lambda,\nc}(s))
    &= -\frac{1}{2\lambda} + \frac{\frac{\partial}{\partial\lambda} R_\lambda(s)}{1 + R_\lambda(s)}
    = -\frac{1}{2\lambda} + O\paren{\frac{q_\lambda}{\lambda^2}}
    = -\frac{1}{2\lambda} + O\paren{\frac{e^{-c/\lambda}}{\lambda^2}} , \notag 
\end{align}
where $ c \coloneqq \pi^2/(4\nc^2) $. 
Recalling the definition \Cref{eqn:Q} of $ Q_\nc $, we then have $ Q_\nc(\lambda) = (2\lambda)^{-1} (1+o(1)) $. 
Therefore, $ Q_\nc(\lambda) \to \infty $ as $ \lambda\downarrow0 $. 
\end{proof}

Since the periodization identity \Cref{eqn:fY} is satisfied by any output density, taking this into account allows us to explicitly solve the output entropy maximization problem \Cref{eqn:entmax} in the next lemma which, due to additivity of channel noise, is equivalent to the mutual information maximization problem in \Cref{eqn:cap}. 
This produces the unique capacity achieving output distribution. 

\begin{lemma}
\label{lem:entmax}
Let $ q>\nc^2/3 $ and let $ \lambda>0 $ be the unique solution to $ Q_\nc(\lambda) = q $. 
Then the following constrained entropy maximization problem 
\begin{align}
    \max\brace{ h(f) : 
        \int y^2 f(y) \diff y \le q , \ 
        \sum_{k\in\bbZ} f(s+2k\nc) = \frac{1}{2\nc} \ \textnormal{ for a.e.\ } s\in [-\nc,\nc)
    } \label{eqn:entmax} 
\end{align}
has a unique maximizer given by $ g_{\lambda,\nc} $ defined in \Cref{eqn:g} and the corresponding maximum entropy equals 
\begin{align}
    h(g_{\lambda,\nc}) &= \lambda q + \log(2\nc) + \frac{1}{2\nc} \int_{-\nc}^\nc \log(\Theta_{\lambda,\nc}(s)) \diff s . \label{eqn:hg} 
\end{align}
\end{lemma}

\begin{proof}
Let $ f $ be a probability density function satisfying the constraints in \Cref{eqn:entmax}. 
By nonnegativity of KL divergence, 
\begin{align}
    0 &\le D(f \Mid g_{\lambda,\nc})
    = \int f(y) \log(f(y)) \diff y - \int f(y) \log(g_{\lambda,\nc}(y)) \diff y . \label{eqn:D} 
\end{align}
This implies 
\begin{align}
    h(f) &\le - \int f(y) \log(g_{\lambda,\nc}(y)) \diff y \notag \\
    &= \lambda \int y^2 f(y) \diff y
    + \log(2\nc)
    + \int f(y) \log(\Theta_{\lambda,\nc}([y]_\nc)) \diff y \label{eqn:1} \\
    &\le \lambda q
    + \log(2\nc)
    + \sum_{k\in\bbZ} \int_{-\nc}^\nc f(s+2k\nc) \log(\Theta_{\lambda,\nc}([s+2k\nc]_\nc)) \diff s \label{eqn:2} \\
    &= \lambda q
    + \log(2\nc)
    + \int_{-\nc}^\nc \paren{ \sum_{k\in\bbZ} f(s+2k\nc) } \log(\Theta_{\lambda,\nc}(s)) \diff s \label{eqn:4} \\
    &= \lambda q
    + \log(2\nc)
    + \frac{1}{2\nc} \int_{-\nc}^\nc \log(\Theta_{\lambda,\nc}(s)) \diff s . \label{eqn:3} 
\end{align}
Here, \Cref{eqn:1} follows since $ \log(g_{\lambda,\nc}(y)) = -\lambda y^2 - \log(2\nc) - \log(\Theta_{\lambda,\nc}([y]_\nc)) $ by the definition \Cref{eqn:g} of $ g_{\lambda,\nc} $; 
in \Cref{eqn:2}, we use the first constraint in \Cref{eqn:entmax} on $ f $ and the assumption $\lambda>0$ to upper bound the first integral from \Cref{eqn:1} by $ \lambda q $, we also split $ \bbR $ into the disjoint union of intervals $ \brace{[s+2k\nc-\nc,s+2k\nc+\nc) : k\in\bbZ} $ for the second integral from \Cref{eqn:1}; 
\Cref{eqn:4} holds since $ [s+2k\nc]_\nc = s $ for any $ s\in[-\nc,\nc) $ and $ k\in\bbZ $; 
\Cref{eqn:3} follows from the second constraint in \Cref{eqn:entmax} on $f$. 

The inequality in \Cref{eqn:D} becomes an equality if and only if $ f = g_{\lambda,\nc} $ almost everywhere. 
To show that the upper bound \Cref{eqn:3} can be achieved, it remains to verify that $ g_{\lambda,\nc} $ is a probability density function satisfying both constraints in \Cref{eqn:entmax}. 
Indeed, 
\begin{align}
    \sum_{k\in\bbZ} g_{\lambda,\nc}(s+2k\nc)
    &= \frac{1}{2\nc} \sum_{k\in\bbZ} \frac{e^{-\lambda(s+2k\nc)^2}}{\Theta_{\lambda,\nc}([s+2k\nc]_\nc)} 
    = \frac{1}{2\nc} \cdot \frac{1}{\Theta_{\lambda,\nc}(s)} \sum_{k\in\bbZ} e^{-\lambda(s+2k\nc)^2}
    = \frac{1}{2\nc} , \notag 
\end{align}
which also implies
\begin{align}
    \int g_{\lambda,\nc}(y) \diff y
    &= \int_{-\nc}^\nc \sum_{k\in\bbZ} g_{\lambda,\nc}(s+2k\nc) \diff s
    = 1 . \notag 
\end{align}
Moreover, 
\begin{align}
    \int y^2 g_{\lambda,\nc}(y) \diff y
    &= \sum_{k\in\bbZ} \int_{-\nc}^\nc (s+2k\nc)^2 g_{\lambda,\nc}(s+2k\nc) \diff s
    = \sum_{k\in\bbZ} \int_{-\nc}^\nc (s+2k\nc)^2 \frac{e^{-\lambda(s+2k\nc)^2}}{2\nc \Theta_{\lambda,\nc}(s)} \diff s
    = Q_\nc(\lambda)
    = q , \notag 
\end{align}
where the last equality is by hypothesis. 
This completes the proof. 
\end{proof}

The following lemma verifies that the pullback (by the channel action) of the output density $ g_{\lambda,\nc} $ obtained in \Cref{lem:entmax} is precisely $ f_X $, thereby identifying the unique capacity achieving input distribution. 

\begin{lemma}
\label{lem:Flam}
Fix $ \lambda>0 $. 
For any $ x\in\bbR $, 
\begin{align}
    F_\lambda(x) &\coloneqq 2\nc \sum_{j=0}^\infty g_{\lambda,\nc}(x - \nc - 2j\nc) \label{eqn:Flambda} 
\end{align}
is an absolutely continuous cumulative distribution function whose corresponding probability density function is given by $ f_X $ defined in \Cref{eqn:fX}. 
Moreover, for any $ y\in\bbR $, 
\begin{align}
    \frac{1}{2\nc} \paren{ F_\lambda(y + \nc) - F_\lambda(y - \nc) }
    &= g_{\lambda,\nc}(y) . \notag 
\end{align}
\end{lemma}

\begin{proof}
For any $x\in\bbR$, there exists a unique pair $ (s,K)\in[-\nc,\nc)\times\bbZ $ such that $ x - \nc = s + 2K\nc $. 
Recalling the notation $ a_k, p_{k,\lambda} $ from \Cref{eqn:pk} and the definition \Cref{eqn:g} of $ g_{\lambda,\nc} $, we have 
\begin{align}
    F_\lambda(x) &= \sum_{j\ge0} \frac{e^{-\lambda(x - \nc - 2j\nc)^2}}{\Theta_{\lambda,\nc}([x - \nc - 2j\nc]_\nc)}
    = \sum_{j\ge0} \frac{e^{-\lambda(s + 2(K - j) \nc)^2}}{\Theta_{\lambda,\nc}([s + 2(K - j)\nc]_\nc)}
    = \sum_{k\le K} \frac{e^{-\lambda(s+2k\nc)^2}}{\Theta_{\lambda,\nc}(s)}
    = \sum_{k\le K} p_{k,\lambda}(s) . \label{eqn:Flam} 
\end{align}
Since $ K\to\pm\infty $ as $ x\to\pm\infty $, it is easy to see that the right-hand side above converges to $ 0 $ as $ x\to-\infty $ and to $ 1 $ as $ x\to\infty $. 

To show continuity of $ F_\lambda $, we only need to consider inputs around the lattice points $ x = 2(K+1)\nc $ for some $ K\in\bbZ $. 
It suffices to show that 
\begin{align}
    \lim_{\eps\downarrow0} F_\lambda(x + \eps) - F_\lambda(x - \eps) &= 0 . \notag 
\end{align}
We decompose $ x-\nc\pm\eps $ into the sum of a lattice point in $ 2\nc\bbZ $ and a residual in $ [-\nc,\nc) $: 
\begin{align}
&&
    x - \nc - \eps &= 2K\nc + (\nc-\eps) , &
    x - \nc + \eps &= 2(K+1)\nc - (\nc-\eps) . & 
& \label{eqn:decomp}
\end{align}
Then by the representation \Cref{eqn:Flam} of $ F_\lambda $, 
\begin{align}
    F_\lambda(x + \eps) - F_\lambda(x - \eps)
    &= \sum_{k\le K+1} p_{k,\lambda}(-(\nc - \eps)) - \sum_{k\le K} p_{k,\lambda}(\nc-\eps) \notag \\
    &= \sum_{k\le K} ( p_{k+1,\lambda}(-(\nc-\eps)) - p_{k,\lambda}(\nc-\eps) ) . \label{eqn:F-F} 
\end{align}
Each term converges to zero as $\eps\downarrow0$: 
\begin{align}
    \lim_{\eps\downarrow0} \, (p_{k+1,\lambda}(-(\nc-\eps)) - p_{k,\lambda}(\nc-\eps))
    &= \lim_{\eps\downarrow0} \paren{ \frac{e^{-\lambda(-(\nc-\eps) + 2(k+1)\nc)^2}}{\Theta_{\lambda,\nc}(-(\nc-\eps) + 2(k+1)\nc)} - \frac{e^{-\lambda(\nc-\eps + 2k\nc)^2}}{\Theta_{\lambda,\nc}(\nc-\eps+2k\nc)} } \notag \\
    &= \lim_{\eps\downarrow0} \paren{ \frac{e^{-\lambda(\nc+\eps + 2k\nc)^2}}{\Theta_{\lambda,\nc}(\nc+\eps + 2k\nc)} - \frac{e^{-\lambda(\nc-\eps + 2k\nc)^2}}{\Theta_{\lambda,\nc}(\nc-\eps+2k\nc)} } 
    = 0 , \label{eqn:p-p} 
\end{align}
by continuity of $ \Theta_{\lambda,\nc}(\cdot) $. 
Therefore, 
\begin{align}
    \lim_{\eps\downarrow0} \, ( F_\lambda(x+\eps) - F_\lambda(x-\eps) ) &= 0 , \notag 
\end{align}
which, in view of \Cref{eqn:F-F}, implies the continuity of $ F_\lambda $. 

We also need to check monotonicity of $ F_\lambda $. 
For $ x \notin 2\nc\bbZ $, there is a unique pair $ (s,K)\in(-\nc,\nc)\times\bbZ $ such that $ x-\nc = s + 2K\nc $. 
In this case, a straightforward calculation shows
\begin{align}
    p_{k,\lambda}'(s)
    &= \frac{-2\lambda a_k(s) e^{-\lambda a_k(s)^2}}{\Theta_{\lambda,\nc}(s)}
    +2 \lambda \frac{e^{-\lambda a_k(s)^2}}{\Theta_{\lambda,\nc}(s)^2} \sum_{j\in\bbZ} e^{-\lambda a_j(s)^2} a_j(s) \notag \\
    &= -2\lambda p_{k,\lambda}(s) \paren{ a_k(s) - \sum_{j\in\bbZ} a_j(s) p_{j,\lambda}(s) } . \notag 
\end{align}
Then by the representation \Cref{eqn:Flam}, we have: 
\begin{align}
    F_\lambda'(x) &= \sum_{k\le K} p_{k,\lambda}'(s)
    = -2\lambda \paren{ \sum_{k\le K} a_k(s) p_{k,\lambda}(s) - \sum_{\substack{k\le K \\ j\in\bbZ}} a_j(s) p_{k,\lambda}(s) p_{j,\lambda}(s) } \notag \\
    &= -2\lambda \paren{ \sum_{\substack{k\le K\\j\in\bbZ}} a_k(s) p_{k,\lambda}(s) p_{j,\lambda}(s) - \sum_{\substack{k\le K \\ j\in\bbZ}} a_j(s) p_{k,\lambda}(s) p_{j,\lambda}(s) }
    \label{eqn:sum1} \\
    &= -2\lambda \sum_{\substack{k\le K \\ j>K}} (a_k(s) - a_j(s)) p_{k,\lambda}(s) p_{j,\lambda}(s)
    \label{eqn:sum0} \\
    &= 4\lambda\nc \sum_{\substack{k\le K\\ j>K}} (j-k) p_{k,\lambda}(s) p_{j,\lambda}(s) , \label{eqn:deriv} 
\end{align}
where \Cref{eqn:sum1} holds since $ \sum_{j\in\bbZ} p_{j,\lambda}(s) = 1 $; 
\Cref{eqn:sum0} follows since 
\begin{align}
    \sum_{\substack{k\le K\\j\le K}} (a_k(s) - a_j(s)) p_{k,\lambda}(s) p_{j,\lambda}(s) &= 0 . \notag 
\end{align}
Note that the right-hand side of \Cref{eqn:deriv} above is obviously nonnegative. 

Next, we show that the derivative at any lattice point $ x = 2(K+1)\nc $ is in fact continuous. 
Again consider $ x-\nc\pm\eps $ which admit the decompositions \Cref{eqn:decomp}. 
By the result \Cref{eqn:deriv} just established, the derivative at $ x-\eps $ is 
\begin{align}
    F_\lambda'(x-\eps)
    &= 4\lambda\nc \sum_{\substack{k\le K \\ j>K}} (j-k) p_{k,\lambda}(\nc-\eps) p_{j,\lambda}(\nc-\eps) , \notag 
\end{align}
and the derivative at $ x+\eps $ is 
\begin{align}
    F_\lambda'(x+\eps)
    &= 4\lambda\nc \sum_{\substack{k\le K+1 \\ j>K+1}} (j-k) p_{k,\lambda}(-(\nc-\eps)) p_{j,\lambda}(-(\nc-\eps)) \notag \\
    &= 4\lambda\nc \sum_{\substack{k\le K \\ j>K}} (j-k) p_{k+1,\lambda}(-(\nc-\eps)) p_{j+1,\lambda}(-(\nc-\eps)) 
    . \notag 
\end{align}
Therefore, by \Cref{eqn:p-p}, 
\begin{align}
    \lim_{\eps\downarrow0} \, (F_\lambda'(x-\eps) - F_\lambda'(x+\eps)) &= 0 . \notag 
\end{align}
Combining this with \Cref{eqn:deriv} and comparing the result with \Cref{eqn:fX}, we conclude that $ F_\lambda $ is continuously differentiable and its derivative is precisely $ f_X $ defined in \Cref{eqn:fX}. 

Finally, we compute
\begin{align}
    \frac{1}{2\nc} (F_\lambda(y+\nc) - F_\lambda(y-\nc))
    &= \sum_{j=0}^\infty g_{\lambda,\nc}(y - 2j\nc) - \sum_{j=0}^\infty g_{\lambda,\nc}(y - 2\nc - 2j\nc) \notag \\
    &= \sum_{j=0}^\infty g_{\lambda,\nc}(y - 2j\nc) - \sum_{j=1}^\infty g_{\lambda,\nc}(y - 2j\nc) 
    = g_{\lambda,\nc}(y) , \notag 
\end{align}
which completes the proof. 
\end{proof}

Equipped with the preceding lemmas, we are ready to complete the proof of \Cref{thm:cap}. 

\begin{proof}[Proof of \Cref{thm:cap}.]
Let $ X $ be any real-valued random variable satisfying $ \expt{X^2} \le \ic $. 
Let $ Y = X + Z_\nc $ for $ Z_\nc \sim \unif([-\nc,\nc)) $ independent of $ X $. 
By \Cref{lem:Y}, the output density $ f_Y $ must satisfy the identity \Cref{eqn:fY} and have second moment at most $ q \coloneqq \ic + \nc^2/3 $. 
Let $ \lambda>0 $ be the unique solution to $ Q_\nc(\lambda) = q $ (existence and uniqueness of solution is guaranteed by \Cref{lem:Q}). 
Then \Cref{lem:entmax} ensures
\begin{align}
    h(Y) &\le \lambda q + \log(2\nc) + \frac{1}{2\nc} \int_{-\nc}^\nc \log(\Theta_{\lambda,\nc}(s)) \diff s . \notag 
\end{align}
It is easy to verify that $ h(Z_\nc) = \log(2\nc) $. 
Therefore, 
\begin{align}
    I(X;Y) &= h(Y) - h(Z_\nc)
    \le \lambda q + \frac{1}{2\nc} \int_{-\nc}^\nc \log(\Theta_{\lambda,\nc}(s)) \diff s . \notag 
\end{align}

On the other hand, we claim that the above upper bound can be attained by the input density $ f_X $ defined in \Cref{eqn:fX}. 
Let $ X_\star \sim f_X $. 
By \Cref{lem:Flam}, the cumulative distribution function of $ X_\star $ is $ F_\lambda $. 
Let $ Z_\nc \sim \unif([-\nc,\nc)) $ be independent of $ X_\star $. 
Then by \Cref{itm:fY} of \Cref{lem:Y}, the density of $ Y_\star \coloneqq X_\star + Z_\nc $ equals 
\begin{align}
    f_{Y_\star}(y) 
    &= \frac{1}{2\nc} \paren{ F_\lambda(y+\nc) - F_\lambda(y-\nc) }
    = g_{\lambda,\nc}(y) , \label{eqn:fY=g} 
\end{align}
where the last equality follows from \Cref{lem:Flam}. 
We also check that $ f_X $ satisfies the power constraint. 
To see this, note that \Cref{lem:entmax} ensures 
\begin{align}
    \expt{Y_\star^2} &= \int y^2 g_{\lambda,\nc}(y) \diff y = q . \notag 
\end{align}
Since $ \expt{Y_\star^2} = \expt{X_\star^2} + \expt{Z_\nc^2} $, this implies $ \expt{X_\star^2} = q-\nc^2/3 = \ic $. 
This proves \Cref{itm:thm3}. 

Combining \Cref{eqn:fY=g,eqn:hg}, we have 
\begin{align}
    I(X_\star;Y_\star)
    &= h(Y_\star) - h(Z_\nc)
    = h(g_{\lambda,\nc}) - \log(2\nc)
    = \lambda q + \frac{1}{2\nc} \int_{-\nc}^\nc \log(\Theta_{\lambda,\nc}(s)) \diff s , \notag 
\end{align}
which proves \Cref{itm:thm2}. 

Finally, in view of \Cref{eqn:fY=g}, uniqueness of the capacity achieving output density $ f_{Y_\star} $ in \Cref{itm:thm1} follows from uniqueness of the maximizer $ g_{\lambda,\nc} $ for \Cref{eqn:entmax}. 
To show uniqueness of $ f_X $, take any capacity achieving input distribution $ F_X $. 
Then its pushforward under the channel action must be the unique output density $ g_{\lambda,\nc} $. 
By \Cref{itm:fY} of \Cref{lem:Y}, 
\begin{align}
    F_X(x) - F_X(x - 2N\nc) &= 2\nc g_{\lambda,\nc}(x - \nc) . \notag 
\end{align}
For any $ N\in\bbZ_{>0} $, consider the telescoping sum: 
\begin{align}
    F_X(x) - F_X(x - 2N\nc)
    &= \sum_{k = 0}^{N-1} F_X(x-2k\nc) - F_X(x-2k\nc-2\nc)
    = 2\nc \sum_{k=0}^{N-1} g_{\lambda,\nc}(x-2k\nc-\nc) . \notag 
\end{align}
Sending $N\to\infty$, we have that $ F_X(x - 2N\nc) $ on the left-hand side above vanishes and the right-hand side tends to $ F_\lambda(x) $ in \Cref{eqn:Flambda}. 
Therefore $ F_X = F_\lambda $, which, by \Cref{lem:Flam}, implies that the density of $ F_X $ must be $ f_X $. 
This completes the proof. 
\end{proof}



\renewcommand*{\bibfont}{\normalfont\small}
\printbibliography


\end{document}